\documentclass[english]{ourlematema}

\usepackage{tikz}
\usepackage{color}

\usepackage{comment}
\usepackage{bbm}
\usepackage{algorithm}
\usepackage{algpseudocode}
\usepackage{listings}
\lstdefinelanguage{Maple}%
{morekeywords={and,assuming,break,by,catch,description,do,done,%
		elif,else,end,error,export,fi,finally,for,from,global,if,%
		implies,in,intersect,local,minus,mod,module,next,not,od,%
		option,options,or,proc,quit,read,return,save,stop,subset,then,%
		to,try,union,use,uses,while,xor},%
	sensitive=true,%
	morecomment=[l]\#,%
	morestring=[b]",%
	morestring=[d]"%
}[keywords,comments,strings]%
\DeclareMathOperator{\Gr}{Gr}

\newcommand{\orcid}[1]{\href{https://orcid.org/#1}{\includegraphics[width=10pt]{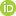}}}

\definecolor{PineGreen}{rgb}{0.0,0.47,0.44}
\definecolor{MidnightBlue}{rgb}{0.1,0.1,0.44}
\definecolor{magenta}{rgb}{1.0,0.0,1.0}
\definecolor{bl1}{HTML}{4479A1}
\definecolor{pur1}{HTML}{52196D}
\definecolor{mag1}{HTML}{2AD0F1}
\definecolor{org1}{rgb}{.92,.39.21}
\definecolor{pur2}{rgb}{.53,.47,.7}
\definecolor{desycyan}{rgb}{0.00,0.68,0.93}
\definecolor{desyorange}{rgb}{0.96,0.52,0.07}
\definecolor{desygray}{rgb}{0.47,0.47,0.47}

\usepackage{hyperref}
\hypersetup{
	colorlinks=true,
	linkcolor=desycyan,
	citecolor=PineGreen,
	filecolor=magenta,
	urlcolor=desyorange
}

\title{The two-loop Amplituhedron} 

\author{Gabriele Dian}
\address{%
	Deutsches Elektronen-Synchrotron DESY\\ 
	Notkestr. 85, 22607 Hamburg, Germany \\
	\email{\hyperlink{gabriele.dian@desy.de}{gabriele.dian@desy.de}}
}
\author{Elia Mazzucchelli}
\address{%
	Max Planck Institute for Physics\\
	Garching bei M\"unchen, Germany\\
	\email{\hyperlink{eliam@mpp.mpg.de}{eliam@mpp.mpg.de}}
}
\author{Felix Tellander}
\address{Mathematical Institute, University of Oxford\\
	Woodstock Road, Oxford OX2 6GG, UK\\
	\email{\hyperlink{felix.tellander@maths.ox.ac.uk}{felix.tellander@maths.ox.ac.uk}}
}

\begin{document}
	\vspace*{-7\baselineskip}%
	\hspace*{\fill} \mbox{\footnotesize{\textsc{MPP-2024-195}}}\newline
	\hspace*{\fill} \mbox{\footnotesize{\textsc{DESY-24-152}}}\newline
	
	\maketitle
	
	\vspace*{4\baselineskip}%
	\begin{abstract}
		The \textit{loop-Amplituhedron} $\mathcal{A}^{(L)}_{n}$ is a semialgebraic set in the product of Grassmannians $\Gr_{\mathbb{R}}(2,4)^L$.  Recently, many aspects of this geometry for the case of $L=1$ have been elucidated, such as its algebraic and face stratification, its residual arrangement and the existence and uniqueness of the adjoint. This paper extends this analysis to the simplest higher loop case given by the two-loop four-point Amplituhedron $\mathcal{A}^{(2)}_4$.
		
	\end{abstract}

	\section{Introduction}

	The \textit{loop Amplituhedron} $\mathcal{A}_{k,n,4,L}$, introduced by Arkani-Hamed and Trnka in \cite{Arkani-Hamed:2013kca}, is a semialgebraic set in $\Gr_{\mathbb{R}}(k,k+4) \times \Gr_{\mathbb{R}}(2,k+4)^L$. It is conjectured that it is a weighted positive geometry \cite{Arkani_Hamed_2017} whose \textit{canonical form} yields the integrand of scattering amplitudes in $\mathcal{N}=4$ super Yang-Mills theory. There, $n$ represents the number of particles, $k$ relates to their helicity, $L$ is the loop order, and we always have $m=4$.
	
	Our work is in fact mainly motivated by positive geometries \cite{Arkani_Hamed_2017}. In that context, the relevant geometric object is a real semialgebraic set, which is equipped with a unique complex meromorphic form called the canonical form. The latter has logarithmic singularities only on its boundaries, and it is holomorphic elsewhere. The definition of a positive geometry \cite[Section 2.1]{Arkani_Hamed_2017} is recursive in the dimension: the residue of the canonical form on each boundary divisor of the semialgebraic set is the canonical form of the boundary, so that the boundary (as a real semialgebraic set) is a positive geometry by itself. In general, it can be difficult to analyze all possible sequences of residues of a canonical form, due to the potentially intricate structure of the boundaries. A natural first step is to analyze the algebraic boundary stratification, and its partition into boundary and residual strata, see Definition \ref{definition_boundary_residual}. Along the lines of \cite{Ranestad:2024svp}, this is the approach we follow in this paper. 
	
	Great progress has been made recently in the study of the boundary structure of tree Amplituhedra, which correspond to the case $L=0$, for $m=2,4$ \cite{Lukowski:2020bya, Even-Zohar:2023del}. For higher loops $L>0$ and $k=0$, $m=4$, the loop Amplituhedron, denoted in the following by $\mathcal{A}^{(L)}_n$, was largely unexplored. A first study of $\mathcal{A}_4^{(L)}$ for $L\leq 4$ has appeared in the physics literature \cite{Franco:2014csa}, where however the presence of \textit{internal boundaries} has not been noticed. Internal boundaries separate regions of opposite orientation and their presence was first observed in \cite{Dian_2023}. Because of this, loop Amplituhedra are not strictly speaking positive geometries \cite{Arkani_Hamed_2017} but are conjectured to be \textit{weighted positive geometries} \cite{Dian_2023}. However, the one-loop case $\mathcal{A}^{(1)}_n$ does not present this feature. In fact, many aspects of $\mathcal{A}^{(1)}_n$ have been elucidated in \cite{Ranestad:2024svp}, including the algebraic and face stratification as well as the existence and uniqueness of the adjoint, which lead to the proof that $\mathcal{A}^{(1)}_n$ is a positive geometry according to the common definition \cite{Arkani_Hamed_2018}. In this paper, we generalise the analysis of \cite{Ranestad:2024svp} and focus on the simplest higher loop case: the two-loop four-point Amplituhedron $\mathcal{A}_4^{(2)}$. 
	
	We determine the stratification of the algebraic boundary of $\mathcal{A}_4^{(2)}$, see Definition \ref{definition_stratum}. Analogously to \cite{Ranestad:2024svp}, we represent all strata as intersections of Schubert varieties, which in this case lie in the product of Grassmannians $\Gr(2,4)^2$. We then determine the intersection of complex strata with $\mathcal{A}_4^{(2)}$ and sort them into boundary and residual strata. Compared to $\mathcal{A}_4^{(1)}$, the real stratification of $\mathcal{A}_4^{(2)}$ exhibits new topological features. The interior is not simply connected and there are boundaries which are disconnected or contain more regions, see Definition \ref{definition_regions}. We checked our results about the algebraic and real stratification computationally, with \texttt{Macauly2} and \texttt{Maple}; our codes are available at \cite{code}. 
	
	Another point we address, generalizing the analysis in \cite{Ranestad:2024svp}, is the adjoint hypersurface of $\mathcal{A}_4^{(2)}$. In the context of Wachspress coordinates and polypols, generalizations of the adjoint have appeared in the literature, see \cite{kohn2021}. In that context, the adjoint is uniquely determined, by requiring it to have a fixed degree and to interpolate the residual arrangement \cite[Proposition 2.2]{kohn2021}. Note that the defining equation of the adjoint curve of a polytope or a polypol is in fact the numerator of its canonical form, see \cite[Theorem 5]{Lam:2022yly} and \cite[Proposition 2.14]{kohn2021}. A definition of the adjoint for general positive geometries is given in \cite[Definition 2]{Lam:2022yly}. Recently, the authors of \cite{Ranestad:2024svp} proved that the one-loop Amplituhedron $\mathcal{A}^{(1)}_n$ has a unique adjoint hypersurface, fully determined by the residual arrangement. We extend this result to the two-loop case $\mathcal{A}^{(2)}_4$.
	
	This paper is structured as follows. Section \ref{sec:Loop_def} defines loop Amplituhedra and provides a geometric way of visualizing $\mathcal{A}^{(2)}_4$. Section \ref{sec:Algebraic} presents our results about the algebraic stratification of $\mathcal{A}_4^{(2)}$. Section \ref{sec:Real} discusses the structure of the real stratification, in particular of the residual arrangement. Section \ref{sec:Topology} is dedicated to topological considerations: we compute the fundamental group of the interior of $\mathcal{A}_4^{(2)}$, the number of connected components and regions, see Definition \ref{definition_regions}, of each (real) boundary stratum. Finally, Section \ref{sec:adjoint} explains how the adjoint geometry of $\mathcal{A}^{(2)}_4$ is uniquely determined by the residual arrangement.

	\section{Loop Amplituhedra and Grassmannians}
	\label{sec:Loop_def}
	
	Conventionally, the Amplituhedron depends on four integer parameters $k, m, n, L$; see for
	instance \cite[Eq. (6.36)]{Arkani_Hamed_2017}. In this paper, we fix the values $k=0$, $m=4$, which turn out to be equivalent to $k = m = 2$ with $L-1$ instead of $L$. The following definition appeared for the first time in \cite{Arkani_Hamed_2018}.
	
	\begin{dfn}[The loop Amplituhedron]
		\label{def_loop_ampl}
		The $n$-point $L$-loop Amplituhedron $\mathcal{A}^{(L)}_n$ for $n \geq 4$ and $L \geq 1$ is the real semialgebraic set in the $L$-fold product of Grassmannians $\Gr_{\mathbb{R}}(2,4)^{L}$ cut out by the following conditions. We fix a totally positive matrix $Z \in \mathbb{R}^{n\times 4}$ and
		represent points in $\Gr_{\mathbb{R}}(2,4)^{L}$ by tuples $(A_1B_1,\dots,A_LB_L)$ of $2\times 4$ matrices $A_{\ell} B_{\ell}$ for $\ell=1,\dots,L$.
		Then, the inequalities cutting out $\mathcal{A}^{(L)}_n$ are given by
		\begin{enumerate}
			\item $\langle A_\ell B_\ell ii+1 \rangle >0$ for $i=1,\dots,n-1$ and $\langle A_\ell B_{\ell} 1n \rangle >0$,
			\item the sequence $(\langle A_{\ell}B_{\ell}12 \rangle, \langle A_{\ell}B_{\ell}13 \rangle, \dots , \langle A_{\ell}B_{\ell}1n \rangle)$ has two sign flips, ignoring zeros,
			\item $\langle A_{\ell}B_{\ell}A_{\ell '}B_{\ell '} \rangle >0$ for every $\ell < \ell '$,
		\end{enumerate}
		for every $\ell, \ell' \in 1,\dots,L$, where we denote by $i$ the $i$-th column of $Z$ and use \textit{twistor coordinates}: for $X,Y \in \Gr_{\mathbb{R}}(2,4)$ we write $\langle XY \rangle$ for the determinant of the $4 \times 4$ matrix obtained by stacking together $X$ with $Y$. Then, we define $\mathcal{A}^{(L)}_n$ to be the Euclidean closure of this semialgebraic set in $\Gr_{\mathbb{R}}(2,4)^L$.
	\end{dfn} 
	We point out that there is an alternative definition of the loop Amplituhedron, see \cite[Eq. (6.35)]{Arkani_Hamed_2017}. In order to introduce it we define the following, which is a special case of \cite[Section 6.4]{Arkani_Hamed_2017}.
	\begin{dfn}[The positive loop Grassmannian]
		\label{def_loop_Grassm}
		We define the $n$-point $L$-loop positive Grassmannian $\Gr_{>0}(2,n;L)$ to be the real semialgebraic set in the $L$-fold product of Grassmannians $\Gr_{\mathbb{R}}(2,n)^{L}$ satisfying the following condition. Let $(D_1,\dots,D_n)$ be matrix representatives of a point in $\Gr_{>0}(2,n)^{L}$. Then, this corresponds to a point in $\Gr_{>0}(2,n;L)$ if for every $1\leq \ell \leq n/2$ and for every $(i_1,\dots,i_{\ell}) \in \binom{[L]}{\ell}$, the $2\ell \times n$ matrix obtained from stacking together $D_{i_1}, \dots , D_{i_{\ell}}$ has positive maximal minors. We define $\Gr_{\geq 0}(2,n;L)$ as the Euclidean closure of $\Gr_{>0}(2,n;L)$ inside $\Gr_{\mathbb{R}}(2,n)^{L}$.
	\end{dfn}
	Then, the $n$-point $L$-loop Amplituhedron is alternatively defined as the image of $\Gr_{\geq 0}(2,n;L)$ under the rational map 
	\begin{equation}\label{Z_projection}
		\tilde{Z} \colon \Gr(2,n)^{L} \rightarrow \Gr(2,4)^{L}, \ ([D_1],\dots,[D_L]) \mapsto ([D_1\cdot Z], \dots, [D_L \cdot Z]) \ ,
	\end{equation}
	where square brackets indicate the equivalence class of a matrix regarded as a point in the Grassmannian.
	To the best of our knowledge, the equivalence of these two definitions is not proven for general $n$ and $L$, but it has been proven in \cite[Lemma 2.3]{Ranestad:2024svp} for $L=1$.
	
	\subsection{Visualizing the two-loop Amplituhedron}
	
	From now on we fix $n=4$. In this case, the map (\ref{Z_projection}) is an isomorphism and we can take without loss of generality $Z$ to be the $4\times 4$ identity matrix. Then, one can check that Definition \ref{def_loop_ampl} and \ref{def_loop_Grassm} are equivalent. The one-loop geometry is just the positive Grassmannian $\mathcal{A}^{(1)}_4=\Gr_{\geq 0}(2,4)$, while what we really are interested in is the two-loop Amplituhedron 
	\begin{equation}\label{two_loop_Ampl}
		\mathcal{A}^{(2)}_4 = \Gr_{\geq 0}(2,4;2) = \{(AB,CD) \in \Gr_{\geq 0}(2,4)^{2} : \langle ABCD \rangle \geq 0 \} \ .
	\end{equation}
	Note that the Grassmannian $\Gr(k,n)$ can be identified with the space of $(k-1)$-planes in $\mathbb{P}^{n-1}$.
	As our notation suggests, we favour this picture, i.e. we identify an element $AB \in \Gr(2,4)$ with a line in $\mathbb{P}^3$ and $A$ and $B$ as distinct points in $\mathbb{P}^3$ lying on $AB$. In this way we can visualize $\mathcal{A}^{(1)}_4 $ in the following way. Let us take a chart of $\mathbb{P}^{3}(\mathbb{R})$ containing the four points $Z_i$ in a compact region and let $AB$ be a line in $\mathbb{P}^3(\mathbb{R})$. By standard projective geometry, we can represent it by any two distinct points on it. We choose
	\begin{equation}\label{A_on_triangle}
		A:=AB \cap (-412)= -Z_4 \langle AB 12 \rangle + Z_1 \langle AB 42 \rangle + Z_2 \langle AB 14 \rangle \ ,
	\end{equation}
	where we denote by $(ij)$ the line through $Z_i$ and $Z_j$, and by $(ijk)$ the plane through $Z_i$, $Z_j$ and $Z_k$. Then,
	the conditions 1 and 2 in Definition \ref{def_loop_ampl} force the coefficients in (\ref{A_on_triangle}) in front of the $Z_i$'s to be positive. This means that $A$ lies in the triangle $T_1$, given by the convex hull of $-Z_4$, $Z_1$ and $Z_2$. Similarly, we chose the point $B$ as
	\begin{equation}\label{B_on_triangle}
		B:=AB \cap (234)= Z_2 \langle AB 34 \rangle + Z_3 \langle AB 42 \rangle + Z_4 \langle AB 23 \rangle \ ,
	\end{equation}
	Again, one finds that the Amplituhedron condition forces $B$ to lie in the triangle $T_2$, given by the convex hull of $Z_2$, $Z_3$ and $Z_4$. Note the necessary choice of the minus sign in front of $Z_4$ in (\ref{A_on_triangle}); a visualization and more details on this can be found in \cite{GabriBlog}. One easily verifies that the geometric picture for the line $AB$ is actually equivalent to the $L=1$ Amplituhedron conditions, i.e. 
	\begin{equation}\label{line_points}
		AB \in \mathcal{A}^{(1)}_4 = \Gr_{\geq 0}(2,4) \iff A \in T_1 \  \text{and} \ B \in T_2 \ .
	\end{equation}
	It is then straightforward to generalize this picture to $\mathcal{A}^{(L)}_4$: we simply introduce $L$ lines, i.e. $L$ points $A_{\ell}$ in $T_1$ and $L$ points $B_{\ell}$ in $T_2$. The complication arises from the condition (3) in Definition \ref{def_loop_ampl}, which imposes relative constraints between pairs of lines $A_{\ell}B_{\ell}$ and $A_{\ell '}B_{\ell '}$. Let us now see how this works for $L=2$. In this case we have two lines, i.e. two pairs of points: $A,C$ in $T_1$ and $B,D$ in $T_2$. These are subject to only one condition given by
	\begin{equation}\label{ABCD}
		\begin{aligned}
			\langle AB CD \rangle =& \langle AB 12 \rangle \langle CD 34 \rangle + \langle AB 34 \rangle \langle CD 12 \rangle + \langle AB 23 \rangle \langle CD 14 \rangle +\\
			&\langle AB 14 \rangle \langle CD 23 \rangle - \langle AB 13 \rangle \langle CD 24 \rangle - \langle AB 24 \rangle \langle CD 13 \rangle \geq  0 \ .
		\end{aligned}
	\end{equation} 
	Note that there is a symmetry under the group $\mathbb{Z}_2 \times D_4$, where $\mathbb{Z}_2$ acts by exchanging $AB$ with $CD$, and the dihedral group $D_4$ of order $4$ acting on the $Z_i$'s. 
	Equation (\ref{ABCD}) can also be rewritten using the Plücker relations for $AB$ and $CD$ as
	\begin{equation}
		\begin{aligned}\label{factorised_ABCD}
			\langle AB CD \rangle = \frac{\Delta^{24}_2 \Delta^{24}_4-\Delta^{24}_1 \Delta^{24}_3 }{\langle AB 24\rangle \langle CD 24 \rangle}  \geq 0 \ ,
		\end{aligned}
	\end{equation}
	where $\Delta^{24}_i$ is defined as
	\begin{equation}\label{Deltas}
		\Delta^{24}_i := \langle AB i i+1 \rangle \langle CD 24 \rangle - \langle AB 24 \rangle \langle CD ii+1 \rangle \ .
	\end{equation}
	We can rewrite (\ref{Deltas}) more geometrically as\footnote{The expression $\langle AB2(CD) \cap (ijk) \rangle$ indicates that the last column of the $4 \times 4$ matrix corresponds to $(CD) \cap (ijk)$.}
	\begin{equation}\label{Delta24_2}
		\Delta_{2}^{24} =   -\langle AB2(CD) \cap (234) \rangle =- \langle AB2D \rangle \ ,
	\end{equation}
	where in the last equality we used (\ref{B_on_triangle}) for $CD$. Similarly, we have\footnote{Note that these expressions are valid only after (\ref{A_on_triangle}) and (\ref{B_on_triangle}).}
	\begin{equation}\label{other_Deltas}
		\Delta_{1}^{24} = -\langle AB2C \rangle \ , \quad \Delta_{3}^{24} = \langle AB4D \rangle \ , \quad \Delta_{1}^{24} = -\langle AB2C \rangle \ .
	\end{equation}
	Equation (\ref{factorised_ABCD}) has the following geometric interpretation: 
	For fixed $A$ and $B$, the sign of (\ref{Delta24_2}) depends only on which side of the line $(2B)$ the point $D$ lies within $T_2$.
	
	From (\ref{other_Deltas}) we can see that the sign regions of the vector $(\Delta^{24}_1,\Delta^{24}_2,\Delta^{24}_3,\Delta^{24}_4)$ triangulate both $T_1$ and $T_2$ into four smaller triangles each.
	Positivity of (\ref{factorised_ABCD}) allows only 12 out of 16 sign-regions. Out of the 16 sign-regions, the following four do not intersect $\mathcal{A}^{(2)}_4$:
	\begin{equation}\label{sign_regions}
		(+---) \,, \quad (+-++) \, , \quad (-+--) \, , \quad (-+++) \ . 
	\end{equation}
	An illustration of this can be found in Figure \ref{two_triangles}.

	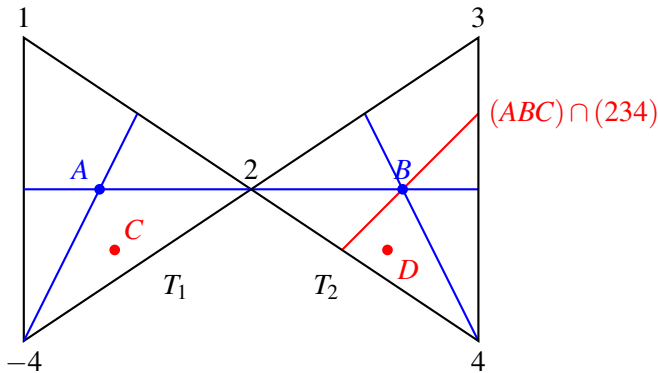
\begin{figure}
		\centering
		\begin{tikzpicture}
			
			\coordinate (1) at (-3,2);
			\coordinate (2) at (0,0);
			\coordinate (-4) at (-3,-2);
			\coordinate (2A) at (-3,0);
			\coordinate (-4A) at (-1.5,1);
			\coordinate (A) at (-2,0);
			\draw[thick] (1) -- (2) -- (-4) -- cycle;
			\draw[blue, thick] (2) -- (2A);
			\draw[blue, thick] (-4) -- (-4A) ;
			\node[above ] at (1) {$1$};
			\node[above] at (2) {$2$};
			\node[below ] at (-4) {$-4$};
			\node[below right ] at (-1.3,-1) {$T_1$};
			\node[above left ] at (A) {{\textcolor{blue}{$A$}}};
			\fill[blue] (A) circle (2pt);

			\coordinate (3) at (3,2);
			\coordinate (4) at (3,-2);
			\coordinate (2B) at (3,0);
			\coordinate (4B) at (1.5,1);
			\coordinate (B) at (2,0);
			\coordinate (P1) at (1.2,-0.8);
			\coordinate (P2) at (3,1);
			\draw[blue, thick] (4) -- (4B);
			\draw[blue, thick] (2) -- (2B);
			\draw[red,thick] (P1) -- (P2);
			\draw[thick] (2) -- (3) -- (4) -- cycle;
			\node[above] at (3) {$3$};
			\node[below ] at (4) {$4$};
			\node[below left ] at (1.3,-1) {$T_2$};
			\node[above] at (B) {{\textcolor{blue}{$B$}}};
			\fill[blue] (B) circle (2pt);

			
			\coordinate (C1) at (-1.8,-0.8);
			\coordinate (C2) at (-1,0.3);
			\node[above right] at (C1) {\textcolor{red}{$C$}};
			\fill[red] (C1) circle (2pt);
			
			\coordinate (D1) at (1.8,-0.8);
			\coordinate (D2) at (1,0.3);
			\node[below right] at (D1) {\textcolor{red}{$D$}};
			\fill[red] (D1) circle (2pt);
			
			\node[right] at (P2) {\textcolor{red}{$(ABC) \cap (234)$}};
		\end{tikzpicture} 
		\caption{Illustration of one point $(AB,CD)$ in $\mathcal{A}^{(2)}_4$ according to (\ref{line_points}). For fixed $A,B,C,$ condition (\ref{ABCD}) fixes on which side of the line $(ABC) \cap(234)$  $D$ lies in $T_2$. The blue lines represent the vanishing of $\Delta^{24}_i$ for $i\in \{1,\dots,4\}$, see (\ref{Delta24_2}) and (\ref{Deltas}).}
		\label{two_triangles}
	\end{figure}

	\section{Algebraic and real boundary stratification}
	\label{sec:Algebraic}

	In this section, we analyze the stratification of the algebraic boundary of the two-loop Amplituhedron and give a full list of complex strata labeled as intersections of Schubert varieties, together with their multiplicities. Finally, we will take the intersection of these strata with the amplituhedron and compute the real boundary stratification.

	We denote the Euclidean boundary of $\mathcal{A}^{(2)}_4$ in $\Gr_{\mathbb{R}}(2,4)^2$ by $\partial \mathcal{A}_4^{(2)}$ and the \textit{algebraic boundary} by $\partial_a \mathcal{A}_4^{(2)}$, which is defined as  the Zariski closure in $\Gr(2,4)^2$ of $\partial \mathcal{A}_4^{(2)}$. 
	Following \cite{Lam:2022yly,Ranestad:2024svp}, we stratify $\partial_a \mathcal{A}^{(2)}_4$ as follows. 
	
	\begin{dfn}\label{definition_stratum}
		We define a \textit{stratum} $S \subset \partial_a \mathcal{A}_4^{(2)}$ to be a complex variety in $\Gr(2,4)^2$ constructed recursively as follows. If the codimension of $S$ is one, then it is one of the irreducible components of $\partial_a \mathcal{A}_4^{(2)}$. If $S$ has codimension $r>1$, then it is an irreducible component of the intersection of two strata of codimension $r-1$ in $\partial \mathcal{A}_4^{(2)}$.
	\end{dfn}
	
	We checked computationally with \texttt{Macaulay2} \cite{M2} that the set of strata is closed under intersections.
	We also checked that all strata are normal. Therefore, one can equivalently define a stratum by iteratively taking the irreducible components of singular loci of codimension one, see \cite[Chapter 9, Theorem 8]{noauthororeditor}.

	\subsection{The boundary divisors}
	\label{section_boundary_divisors}

	We first determine the \textit{boundary strata}, i.e. the irreducible components of $\partial_a \mathcal{A}_4^{(2)}$. 
	By (\ref{two_loop_Ampl}), $\partial_a \mathcal{A}_4^{(2)}$ will contain the algebraic stratification of two copies of $\mathcal{A}^{(1)}_4 $, which one can find in \cite[Table 1]{Ranestad:2024svp}. The elements in the stratification of $\mathcal{A}^{(1)}_4 $ can be expressed as intersections of the following Schubert varieties:{\small
		\begin{equation}\label{Schubert_divisors}
			\begin{aligned}
				L^{(1)}_{i} :=\{AB \cap (ii+1) \neq \emptyset \} \, , \
				V^{(1)}_i := \{i \in AB \} \, , \ 
				P^{(1)}_i := \{AB \subset (i-1ii+1) \} \, ,
			\end{aligned}
	\end{equation}}
	for $i\in\{1,\dots,4\}$ taken modulo $4$. Analogously, we use the superscript $(2)$ for the corresponding varieties for $CD$. We refer to elements in $\partial_a \mathcal{A}_4^{(2)}$ that can be written as intersections of only these Schubert varieties as \textit{product strata}, and to their union as \textit{product stratification}. In this way, by (\ref{two_loop_Ampl}) there are eight boundary divisors in the product stratification of $\mathcal{A}_4^{(2)}$, given by $L_i^{(\ell)}$ for $i=1,\dots,4$ and $\ell = 1,2$.
	There is an additional bidegree $(1,1)$ boundary divisor given by the vanishing of (\ref{ABCD}), which we denote by
	\begin{equation}\label{AB_CD_divisor}
		L(1,2) := \{AB,CD \subset \mathbb{P}^3 : AB \cap CD \neq \emptyset \} \subset \Gr(2,4)^2 \ .
	\end{equation}
	There are therefore nine boundary divisors\footnote{This fact can be easily checked by choosing a parametrisation of $\Gr_{\mathbb{R}}(2,4)^2$.}.
	We point out that the singular locus of $L_i^{(1)}$ and that of $L(1,2)$ have both codimension three and are given by $AB = (ii+1)$ and $\{AB=CD\}$ respectively. Therefore, all boundary divisors are normal algebraic varieties.
	
	The more intricate part of the determination of the stratification of $\partial_a \mathcal{A}_4^{(2)}$ lies in the strata involving an intersection with (\ref{AB_CD_divisor}). For instance, on $L^{(1)}_1 L^{(2)}_1$ (where from now on we omit the intersection symbols between Schubert varieties) we have that (\ref{factorised_ABCD}) factorises as
	\begin{equation}\label{Li_Li_factorisation}
		\langle ABCD \rangle = \frac{\Delta^{24}_{2} \Delta^{24}_{4}}{\langle AB 24\rangle \langle CD 24 \rangle} \ ,
	\end{equation}
	which consists of two irreducible components, both of codimension three.  In terms of Shubert varieties, we will denote this as 
	\begin{equation}\label{L12_factroization}
		L(1,2) L^{(1)}_1 L^{(2)}_1  = V_{1}(1,2)  \cup P_1(1,2) \ 
	\end{equation}
	where the component $V_{1}(1,2)$ is defined by the vanishing of $\Delta_{4}^{24}$ and corresponds to the geometric configuration in which the three lines $AB$, $CD$ and $(12)$ intersect in a single point, while $P_1(1,2)$ is defined by the vanishing of $\Delta_{2}^{24}$, denotes the component on which they are coplanar. 
	
	Other relevant instances of factorisation are 
	\begin{equation}\label{factorization_V12}
		\begin{aligned}
			V_{i-1}(1,2)V_i(1,2) &= V_i^{(1)}V_i^{(2)} \cup P_i^{(12)}  \ , \\
			P_{i-1}(1,2)P_i(1,2) &= P_i^{(1)}P_i^{(2)} \cup V_i^{(12)}  \ ,
		\end{aligned}
	\end{equation}
	where $V_i^{(12)}$ denotes the two-dimensional irreducible component where $AB=CD$ passes through $Z_i$, while $P_i^{(12)}$ denotes the two-dimensional component where $AB=CD$ lies in the plane $(i-1ii+1)$. Note that the irreducible components on the right-hand side of (\ref{factorization_V12}) have dimension four and two respectively. This is due to the fact that the divisors we are considering are non-generic. 
	
	\subsection{Strata in each codimension}\label{section_strata}

	We now present our result about the computation of all strata in $\partial_a \mathcal{A}_4^{(2)}$. We computed all strata using the Schubert-like relations listed in Section \ref{sec:Schubert} and checked our result using \texttt{Macaulay2} \cite{M2}.
	A summary of the counting of strata ordered by codimension is given in Table \ref{tab: algebraic strat two-loop}.  
	We give a list of all strata in Appendix \ref{appendix_lists}.
	\begin{table}[h!]
		\centering
		\begin{tabular}{l|c|c|c|c|c|c|c|c}
			Codimension: & 1 & 2 & 3 & 4 & 5 & 6 & 7 & 8 \\
			\hline
			Components: & 9 & 44 & 144 & 324 & 450 & 370 & 168 & 36 \\
			\hline
			Boundaries: & 9 & 44 & 144 & 286 & 356 & 306 & 156 & 34 \\
			\hline
			Residual: & 0 & 0 & 0 &  38 & 94 & 64 & 12 & 2 \\
			\hline
			Regions: & 9 & 52 & 176 & 326 & 416 & 342 & 156 & 34
		\end{tabular}
		\caption{Number of strata in each codimension of the algebraic boundary of $\mathcal{A}^{(2)}_4$, with a breakdown of boundary and residual components, as well as the number of regions, see Definition \ref{definition_regions}. }
		\label{tab: algebraic strat two-loop}
	\end{table}
	Our results are summarized as follows. We present one table for each codimension, ranging from one to eight. The indices run as follows: $\ell =1,2$ is taken modulo $2$ and $i,j,k,l =1,\dots,4$ are taken modulo $4$. Each stratum's multiplicity is given, along with its classification as either part of the boundary or residual arrangement. The next column shows the number of connected components formed by the intersection with $\mathcal{A}^{(2)}_{4}$, while the final column lists the number of regions for each boundary stratum, both of which are discussed in Section \ref{section_CW_structure}.  Additionally, each table features a horizontal separation into two families: strata above the line belong to the product stratification, while those below do not, and are therefore contained in $L(1,2)$.

	\subsection{Schubert-like relations}
	
	\label{sec:Schubert}
	By standard results in Schubert calculus, the varieties in (\ref{Schubert_divisors}) satisfy 
	\begin{equation}\label{1_loop_relations}
		L_{i-1}^{(\ell)}L_{i}^{(\ell)}=V_i^{(\ell)}\cup P_i^{(\ell)} \  , \quad  V_i^{(\ell)}L_{i+1}^{(\ell)}= V_i^{(\ell)}P_{i+1}^{(\ell)} \ , \quad  V_i^{(\ell)}L_{i-2}^{(\ell)}= V_i^{(\ell)}P_{i-1}^{(\ell)} \ .
	\end{equation}
	They also satisfy (\ref{L12_factroization}) and (\ref{factorization_V12}). Additionally, there are other relations. One way of finding these is to consider one of the lines to be fixed and using standard Schubert calculus with respect to the other one. Here we give a (conjecturally minimal) list of relations.
	First of all, we have product strata which lie in $L(1,2)$:
	\begin{equation}\label{L_12_invariant}
		V^{(\ell)}_i V^{(\ell+1)}_i \ , \qquad P^{(\ell)}_i P^{(\ell+1)}_i \ ,  \qquad V^{(\ell)}_i V^{(\ell)}_{i+1}L^{(\ell+1)}_i \ , \qquad L^{(\ell)}_i V^{(\ell+1)}_i V^{(\ell+1)}_{i+1} \ ,
	\end{equation}
	and all the substrata thereof. Secondly, we find the following relations:
	\begin{equation}\label{2_loop_relations}
		\begin{aligned}
			V_i(1,2)V_{i}^{(\ell)} &= V_{i}^{(\ell)} V_{i}^{(\ell+1)} \cup V^{(\ell)}_{i} V^{(\ell)}_{i+1} L^{(\ell+1)}_{i} \ , \\  P_i(1,2)P_{i}^{(\ell)} &= P_{i}^{(\ell)} P_{i}^{(\ell+1)} \cup V^{(\ell)}_{j} V^{(\ell)}_{i+1} L^{(\ell+1)}_{i} \ , \\
			V_j(1,2)V_{i}^{(\ell+1)} V_{i}^{(\ell)} &=  V_{i}^{(\ell)} V_{i}^{(\ell+1)} \ , \  P_j(1,2)P_{i}^{(\ell)} P_{i}^{(\ell+1)} = P_{i}^{(\ell)} P_{i}^{(\ell+1)} \quad j = i,i-1  \ , \\
			P_i(1,2) V^{(\ell)}_{i} V^{(\ell)}_{i+1}  &= V_i(1,2) V^{(\ell)}_{i} V^{(\ell)}_{i+1} =  V^{(\ell)}_{i} V^{(\ell)}_{i+1} L^{(\ell+1)}_{i}  \ , \\
			P_i(1,2) P_{i+2}(1,2) &= V_i(1,2) V_{i+2}(1,2) \ , \\
			V_i(1,2) P_{i+2}(1,2) &= V_i(1,2) L_{i+2}^{(\ell)}L^{(\ell+1)}_{i+2} =  P_{i+2}(1,2) L_{i}^{(\ell)}L^{(\ell+1)}_{i} \ , \\ 
			V_i(1,2) P_{i+1}(1,2) &= V_i(1,2) P^{(\ell)}_i P^{(\ell+1)}_i \cup P_{i+1}(1,2)V_{i}^{(\ell)}V_{i}^{(\ell+1)} \ , \\
			V^{(12)}_i P^{(12)}_i &= V^{(12)}_i P^{(\ell)}_i  = V^{(12)}_i P^{(\ell+1)}_i = P^{(12)}_i V^{(\ell)}_i = P^{(12)}_i V^{(\ell+1)}_i \ , \\ 
			V_i^{(12)} P^{(\ell)}_{i+1} P^{(\ell+1)}_{i+1} &= V_{i+1}(1,2) V^{(\ell)}_i V^{(2+1)}_i \ .
		\end{aligned}
	\end{equation}

	\subsection{Real stratification}
	\label{sec:Real}
	
	We now determine the real boundary stratification of $\mathcal{A}^{(2)}_4$. We introduce Definition \ref{definition_boundary_residual} which is motivated by the discussion we present in Section~\ref{sec:adjoint}. 
	\begin{dfn}\label{definition_boundary_residual}
		Let $S$ be a complex stratum in $\partial_a \mathcal{A}^{(2)}_4$, see Definition \ref{definition_stratum}. We write $S_{\geq 0} := S \cap \mathcal{A}^{(2)}_4$ and denote by $S_{>0}$ the relative interior of $S_{\geq 0}$ in the real points $S(\mathbb{R})$ of $S$. Note that $S_{>0}$ is a real manifold. We call $S$ a \textit{boundary} or \textit{face} of $\mathcal{A}^{(2)}_4$ if the real dimension of $S_{>0}$ is equal to the complex dimension of $S$, and \textit{residual} otherwise. We call \textit{residual arrangement} of $\mathcal{A}^{(2)}_4$ the (projective) variety given by the union of all residual strata.
	\end{dfn}%
	
	We now focus our attention on the determination of the residual arrangement of $\mathcal{A}^{(2)}_4$. 
	Recall that the residual arrangement of $\mathcal{A}^{(1)}_4$ is empty. Instead, for $\mathcal{A}^{(2)}_4$ the residual arrangement is 4-dimensional and its 4-dimensional irreducible components are
	\begin{center}
		\begin{tikzpicture}[]

			\begin{scope}[xshift=0 cm]
				\node at (0.5,-1) {$V_1^{(1)}V_3^{(1)} $};
				\draw[dashed] (0,0) -- (0,1) -- (1,1) -- (1,0) -- cycle;
				
				\foreach \x/\y/\n in {0/0/1, 1/0/2, 1/1/3, 0/1/4}{
					\node[draw,circle,fill=black,inner sep=0pt,minimum size=4pt] at (\x,\y) {};
				}
				
				\node[draw,rectangle,fill=blue,inner sep=0pt,minimum size=6pt] at (0,0) {};
				
				\node[draw,rectangle,fill=blue,inner sep=0pt,minimum size=6pt] at (1,1) {};
				
				\node at (0,0) [below left] {1};
				\node at (0,1) [above left] {2};
				\node at (1,1) [above right] {3};
				\node at (1,0) [below right] {4};
				
			\end{scope}
			
			
			\begin{scope}[xshift=2.4cm]

				\node at (0.5,-1) {$V_1^{(1)}P_2^{(1)}L_1^{(2)} $};
				\draw[dashed] (0,0) -- (0,1) -- (1,1) -- (1,0) -- cycle;
				\fill[blue, opacity=0.5] (0,0) -- (0,1) -- (1,1) -- cycle;
				
				\foreach \x/\y/\n in {0/0/1, 1/0/2, 1/1/3, 0/1/4}{
					\node[draw,circle,fill=black,inner sep=0pt,minimum size=4pt] at (\x,\y) {};
				}
				
				\node[draw,rectangle,fill=blue,inner sep=0pt,minimum size=6pt] at (0,0) {};
				
				\node[draw,rectangle ,fill=red,inner sep=2pt,minimum size=6pt] at (0,0.5) {};

				\node at (0,0) [below left] {1};
				\node at (0,1) [above left] {2};
				\node at (1,1) [above right] {3};
				\node at (1,0) [below right] {4};
				
			\end{scope}
			
			
			\begin{scope}[xshift=4.8cm]

				\node at (0.5,-1) {$P_2^{(1)}V_2^{(2)} $};
				\draw[dashed] (0,0) -- (0,1) -- (1,1) -- (1,0) -- cycle;
				\fill[blue, opacity=0.5] (0,0) -- (0,1) -- (1,1) -- cycle;
				
				\foreach \x/\y/\n in {0/0/1, 1/0/2, 1/1/3, 0/1/4}{
					\node[draw,circle,fill=black,inner sep=0pt,minimum size=4pt] at (\x,\y) {};
				}
				
				
				\node[draw,rectangle ,fill=red,inner sep=2pt,minimum size=6pt] at (0,1) {};

				\node at (0,0) [below left] {1};
				\node at (0,1) [above left] {2};
				\node at (1,1) [above right] {3};
				\node at (1,0) [below right] {4};
				
			\end{scope}
			
			
			\begin{scope}[xshift=7.2cm]

				\node at (0.5,-1) {$L_1^{(1)}L_3^{(1)}L_2^{(2)}L_4^{(2)} $};
				\draw[dashed] (0,0) -- (0,1) -- (1,1) -- (1,0) -- cycle;

				\foreach \x/\y/\n in {0/0/1, 1/0/2, 1/1/3, 0/1/4}{
					\node[draw,circle,fill=black,inner sep=0pt,minimum size=4pt] at (\x,\y) {};
				}
				
				
				\node[draw,rectangle ,fill=red,inner sep=2pt,minimum size=6pt] at (0.5,0) {};
				\node[draw,rectangle ,fill=red,inner sep=2pt,minimum size=6pt] at (0.5,1) {};
				\node[draw,rectangle ,fill=blue,inner sep=2pt,minimum size=6pt] at (0,0.5) {};
				\node[draw,rectangle ,fill=blue,inner sep=2pt,minimum size=6pt] at (1,0.5) {};

				\node at (0,0) [below left] {1};
				\node at (0,1) [above left] {2};
				\node at (1,1) [above right] {3};
				\node at (1,0) [below right] {4};
				
			\end{scope}

		\end{tikzpicture}
	\end{center}
	\begin{center}
		\begin{tikzpicture}[]

			\begin{scope}[xshift=0cm]

				\node at (0.6,-1) {$L(1,2)V_2^{(1)}P_2^{(1)} $};
				\draw[dashed] (0,0) -- (0,1) -- (1,1) -- (1,0) -- cycle;
				\fill[blue, opacity=0.5] (0,0) -- (0,1) -- (1,1) -- cycle;
				
				\foreach \x/\y/\n in {0/0/1, 1/0/2, 1/1/3, 0/1/4}{
					\node[draw,circle,fill=black,inner sep=0pt,minimum size=4pt] at (\x,\y) {};
				}
				
				\node[draw,rectangle,fill=blue,inner sep=0pt,minimum size=6pt] at (0,1) {};
				
				\node[draw,rectangle ,fill=red,inner sep=2pt,minimum size=6pt] at (0.3,0.7) {};

				\node at (0,0) [below left] {1};
				\node at (0,1) [above left] {2};
				\node at (1,1) [above right] {3};
				\node at (1,0) [below right] {4};
				
			\end{scope}


			\begin{scope}[xshift=3 cm]

				\node at (0.6,-1) {$L(1,2)V_1^{(1)}P_2^{(1)} $};
				\draw[dashed] (0,0) -- (0,1) -- (1,1) -- (1,0) -- cycle;
				\fill[blue, opacity=0.5] (0,0) -- (0,1) -- (1,1) -- cycle;
				
				\foreach \x/\y/\n in {0/0/1, 1/0/2, 1/1/3, 0/1/4}{
					\node[draw,circle,fill=black,inner sep=0pt,minimum size=4pt] at (\x,\y) {};
				}
				
				\node[draw,rectangle,fill=blue,inner sep=0pt,minimum size=6pt] at (0,0) {};
				
				\node[draw,rectangle ,fill=red,inner sep=2pt,minimum size=6pt] at (0.3,0.7) {};

				\node at (0,0) [below left] {1};
				\node at (0,1) [above left] {2};
				\node at (1,1) [above right] {3};
				\node at (1,0) [below right] {4};
				
			\end{scope}

		\end{tikzpicture}
	\end{center}
	as well as all the configurations in their orbits under the symmetry group, see Table \ref{table_codim_four}. Here we follow the pictorial depiction of Schubert conditions introduced in \cite{Ranestad:2024svp}. The bracket $\langle ABCD \rangle$ reduces to $-\langle AB24 \rangle \langle CD 13 \rangle$ in the first three configurations and $-\langle AB24 \rangle \langle CD 13 \rangle - \langle AB13 \rangle \langle CD 24 \rangle $ in the fourth one. Equation (\ref{ABCD}) implies that the intersection of these complex strata with $\mathcal{A}^{(2)}_4$ is empty. On the variety $L(1,2)V_2^{(1)}P_2^{(1)}$  we have instead
	\begin{equation}\label{residual_by_drop}
		\langle ABCD \rangle =\langle AB 34 \rangle \langle CD 12 \rangle + \langle AB 14 \rangle \langle CD 23 \rangle = 0 \ .
	\end{equation}
	It follows that the intersection of this four-dimensional stratum with $\mathcal{A}^{(2)}_4$ is three-dimensional and therefore residual. An analogous situation happens to $L(1,2)V_1^{(1)}P_2^{(1)}$, as well as to the following three-dimensional  complex strata:
	\begin{center}
		\begin{tikzpicture}[]

			\begin{scope}[xshift=0 cm]
				\node at (0.5,-1) {$L(1,2)V_1^{(1)}V_3^{(2)} $};
				\draw[dashed] (0,0) -- (0,1) -- (1,1) -- (1,0) -- cycle;
				
				\foreach \x/\y/\n in {0/0/1, 1/0/2, 1/1/3, 0/1/4}{
					\node[draw,circle,fill=black,inner sep=0pt,minimum size=4pt] at (\x,\y) {};
				}
				
				\node[draw,rectangle,fill=blue,inner sep=0pt,minimum size=6pt] at (0,0) {};
				
				\node[draw,rectangle,fill=red,inner sep=0pt,minimum size=6pt] at (1,1) {};
				
				\node[draw,rectangle,fill=blue,inner sep=0pt,minimum size=6pt] at (0.45,0.45) {};
				
				\node[draw,rectangle,fill=red,inner sep=0pt,minimum size=6pt] at (0.55,0.55) {};
				
				\node at (0,0) [below left] {1};
				\node at (0,1) [above left] {2};
				\node at (1,1) [above right] {3};
				\node at (1,0) [below right] {4};
				
			\end{scope}
			
			
			\begin{scope}[xshift=4cm]

				\node at (0.5,-1) {$L(1,2)P_1^{(1)}P_3^{(2)} $};
				\draw[dashed] (0,0) -- (0,1) -- (1,1) -- (1,0) -- cycle;
				\fill[blue, opacity=0.5] (1,0) -- (0,0) -- (0,1) -- cycle;
				\fill[red, opacity=0.5] (1,0) -- (1,1) -- (0,1) -- cycle;
				
				\foreach \x/\y/\n in {0/0/1, 1/0/2, 1/1/3, 0/1/4}{
					\node[draw,circle,fill=black,inner sep=0pt,minimum size=4pt] at (\x,\y) {};
				}
				
				\node[draw,rectangle,fill=blue,inner sep=0pt,minimum size=6pt] at (0.45,0.45) {};
				
				\node[draw,rectangle,fill=red,inner sep=0pt,minimum size=6pt] at (0.55,0.55) {};

				\node at (0,0) [below left] {1};
				\node at (0,1) [above left] {2};
				\node at (1,1) [above right] {3};
				\node at (1,0) [below right] {4};
				
			\end{scope}

		\end{tikzpicture}
	\end{center}
	All other strata of the residual arrangement are contained in the ones above and can be found in Appendix \ref{appendix_lists}.
	The residual arrangement can also be calculated directly using the \texttt{RegularChains} \cite{lemaire2005regularchains} library in \texttt{Maple}. We used \texttt{LazyRealTriangularize} \cite{Chen2013} to determine the residual arrangement and the inequality description of boundaries' interior.

	\section{Topology}
	\label{sec:Topology}

	Now we turn our attention to the study of the real topology of $\mathcal{A}^{(2)}_4$. While the interior of $\mathcal{A}^{(1)}_4$ is homeomorphic to an open ball \cite{Galashin_2022}, this is no longer true at two loops.
	
	\begin{thm}\label{theorem_topology}
		The Euclidean interior ${\rm int}(\mathcal{A}^{(2)}_4)$ of $\mathcal{A}^{(2)}_4$ is connected and its fundamental group is free of rank one\footnote{We thank Thomas Lam for pointing this out to us.}.
	\end{thm}
	\begin{proof}
		Let us start by observing that ${\rm int}(\mathcal{A}^{(2)}_4)$ is a fiber bundle over ${\rm int}(\mathcal{A}^{(1)}_4)$ via the projection map on the first component $AB$, with fibers homeomorphic to 
		\begin{equation}\label{fiber}
			\mathcal{F}:= \{CD \in {\rm int}(\mathcal{A}^{(1)}_4) : \langle AB CD  \rangle > 0 \} \ ,
		\end{equation}
		for any $AB \in {\rm int}(\mathcal{A}^{(1)}_4)$. Since ${\rm int}(\mathcal{A}^{(1)}_4)$ is contractible \cite{Galashin_2022}, the fiber bundle is trivial and therefore, if $\mathcal{F}$ is connected so is ${\rm int}(\mathcal{A}^{(2)}_4)$. Moreover, the fundamental group of ${\rm int}(\mathcal{A}^{(2)}_4)$ is isomorphic to that of $\mathcal{F}$. 
		
		Let us start by showing that $\mathcal{F}$ is connected. We want to find a path between two points $C_iD_i \in \mathcal{F}$ for $i=1,2$. We choose  $A$, $C_i$ to lie on $T_1$ and $B$, $D_i$ on $T_2$, see Figure~\ref{two_triangles}. A path $C(t)D(t)$ with $t \in [0,1]$ from $C_1D_1$ to $C_2D_2$ draws paths $C(t)$ in $T_1$ and $D(t)$ in $T_2$. Let us take a path $C(t)$ connecting $C_1$ to $C_2$. For each point $C(t)$ we can see that the allowed region for $D(t)$ is determined by $\langle ABC(t)D(t) \rangle >0$, which carves out one of the two connected components of the complement of $T_2$ by the line $L(t):=(ABC(t))\cap(234)$. The latter is drawn in red in Figure~\ref{two_triangles}. By continuity there always exists a path $D(t)$ from $D_1$ to some point $D(1)$ such that $C(t)D(t)$ always lies in $\mathcal{F}$. It is now easy to see that $D(1)$ and $D_2$ lie in the same connected component of the complement of $T_2$ by $L(1)$, and hence $C_2 D(1)$ and $C_2 D_2$ can be connected by a path in $\mathcal{F}$.
		
		We now study at the fundamental group of $\mathcal{F}$.
		As before, a loop $\gamma(t)=C(t)D(t)$ in $\mathcal{F}$ draws loops $C(t)$ in $T_1$ and $D(t)$ in $T_2$. A homotopy between loops in $\mathcal{F}$ yields homotopies between the respective loops in $T_1$ and $T_2$. Assume that $C(t)$ has winding number $w\in \mathbb{Z}$ around $A$. Since $D(t)$ must always lie on the same side of $L(t)$ defined as in the previous paragraph, it is clear that the winding of $D(t)$ around $B$ must be $w$ as well. Finally, if $w \neq 0$ then $\gamma(t)$ is not contractible, because $C(t)$, $D(t)$ cannot be equal to $A$, $B$ respectively because of (\ref{fiber}). 
	\end{proof}

	\subsection{Towards the CW-structure}\label{section_CW_structure}

	On the positive Grassmannian there is a regular CW structure given by the positroid stratification \cite{postnikov2006}. More precisely, each $d$-dimensional (complex) stratum in $\partial_a \mathcal{A}^{(1)}_4$ intersects $\mathcal{A}^{(1)}_4$ in a (real) $d$-dimensional set, which is in fact a (closed) positroid cell. We can therefore endow $\mathcal{A}^{(1)}_4$ with a regular CW structure directly from the algebraic stratification. 
	
	When considering $\mathcal{A}^{(2)}_4$, there are three main aspects we want to highlight: the nontrivial topology of the interior, the non-connectedness of some faces (see also \cite{Franco:2014csa}) and the presence of multiple regions, see Definition \ref{definition_regions}.  Finally, we explain how one may give a CW structure to $\mathcal{A}^{(2)}_4$ starting from its face stratification.
	
	Let us begin with the first point. By Theorem \ref{theorem_topology}, to equip $\mathcal{A}^{(2)}_4$ with a CW structure, we must divide its interior into at least two cells. The following lemma demonstrates that two cells suffice, as shown by a straightforward application of the argument in the proof of Theorem \ref{theorem_topology}.

	\begin{lemma}\label{lemma_top_cells}
		The two regions in ${\rm int}(\mathcal{A}^{(2)}_4)$ defined by $\{ \Delta_i^{24}>0 \}$ and $\{ \Delta_i^{24}<0 \}$ for any fixed $i \in \{1,\dots,4\}$ are both contractible.
	\end{lemma}

	The second distinctive feature at two loops is the presence of disconnected faces. Their number is nevertheless very limited and they appear only at codimension two and three, see Table \ref{table_codim_two} and \ref{table_codim_three}. Finding the connected components of a real semialgebraic set is a famously challenging problem. It is in principle solved by calculating a \emph{cylindrical algebraic decomposition} (CAD) \cite{Collins1975}. The CAD yields a decomposition into disjoint cells that are ``cylindrical" (see \cite{Collins1975} for precise meaning). Each cell is homeomorphic to an open cube and therefore connected. If top-dimensional cells are connected by lower-dimensional cells, they are a part of the same connected component. We used the CAD implementation in the \texttt{Maple} library \texttt{RegularChains} \cite{lemaire2005regularchains} for determining the number of connected components of boundaries. We now present the disconnected boundaries of $\mathcal{A}^{(2)}_4$. In the following we denote a face by the same expression as that of the stratum it is associated to. Then, the only disconnected faces at codimension two are those in the orbit of $L^{(1)}_1L^{(1)}_3$, on which (\ref{factorised_ABCD}) reduces to
	\begin{equation}
		\langle ABCD \rangle = \frac{\Delta^{24}_2 \Delta^{24}_4}{\langle AB24 \rangle \langle CD24\rangle} - \frac{\langle AB24\rangle}{\langle CD24 \rangle} \langle CD12\rangle \langle CD34\rangle \geq 0 \ .
	\end{equation}
	In the plane $(\Delta^{24}_2,\Delta^{24}_4)$ the face $L^{(1)}_1L^{(1)}_3$ consists of two connected components bounded by the hyperbola $\langle ABCD \rangle =0$, which corresponds to $L(1,2)L^{(1)}_1L^{(1)}_3$. The face associated to the latter also has two connected components, corresponding to the two branches of the hyperbola. It turns out that these are the only disconnected (closed) faces. There are however faces whose interior is disconnected, which we also discuss now.
	
	The last point we want to address is that of multiple regions, which is tightly related to the presence of \textit{internal boundaries} \cite{Dian_2023}.
	\begin{dfn}\label{definition_regions}
		We define a \textit{region} of a $d$-dimensional face $S_{\geq 0}$ with $d>0$ to be a connected component of the complement in $S_{\geq 0}$ by the union of all $(d-1)$-dimensional (closed) faces.
	\end{dfn}

	We now give some examples. Consider the codimension two face $L^{(1)}_1L^{(2)}_1$, on which $\langle ABCD \rangle \geq 0$ reduces to $\Delta^{24}_2 \Delta^{24}_4 \geq 0$. Its interior has two connected components which can be visualized in the plane $(\Delta^{24}_2,\Delta^{24}_4)$ as the union of the two sign regions $(++)$ and $(--)$. Its boundary is given by the intersection with $L(1,2)$ and yields the union of the two coordinate axis, $\Delta^{24}_2 \Delta_{4}^{24}=0$, which are the faces $P_1(1,2)$ and $V_1(1,2)$. The latter are further separated by the boundary $P_1(1,2)V_1(1,2)$, which corresponds to the origin $\Delta^{24}_2 = \Delta_{4}^{24}=0$. Therefore, $L^{(1)}_1L^{(2)}_1$, $P_1(1,2)$ and $V_1(1,2)$ all have two regions.
	
	As the second example, we take the codimension four face $V^{(1)}_1V^{(2)}_1$. On this we have $\langle ABCD \rangle = 0$
	but $\Delta^{24}_2, \Delta^{24}_3 \neq 0$. In particular, if we look at the plane $(\Delta^{24}_2, \Delta^{24}_3)$, we see that the face consists of all four sign regions. The face has four regions, separated by two higher codimension boundaries corresponding to the two coordinate axes, which are $P_i(1,2)V^{(1)}_1V^{(2)}_1 $ and for $i=2,4$ respectively. The latter are in turn separated by  $V_1^{(12)}$ into two regions each. An analogous discussion applies for the face $P^{(1)}_1P^{(2)}_1$.
	
	The highest codimension faces with multiple regions have codimension six. For example the face $P_4(1,2)V_1^{(1)}P_2^{(1)}V_1^{(2)}$ is separated by the $V_1^{(12)}P_2^{(1)}P_2^{(2)}$ into two regions, characterised by the sign of $\Delta^{24}_3$.
	
	These examples show that in order to endow $\mathcal{A}^{(2)}_4$ with a regular CW structure from its face stratification, one has to subdivide some faces into several cells. Let us mention the minimal requirements for such a procedure. There must be at least two eight-dimensional cells according to Lemma \ref{lemma_top_cells}. The same applies to the disconnected faces presented above. The zero-dimensional cells are the 34 vertices of Table \ref{table_codim_eight}. All other faces must be decomposed in at least as many (open) cells as their regions. We find that all these can be characterised by the sign regions of $\Delta^{24}_i$ or $\Delta^{13}_i$ for $i=1,\dots,4$, see (\ref{Deltas}). We enumerate the regions of each face in the tables of Section \ref{appendix_lists}. All regions are easily verified to be contractible\footnote{One can also analyze this through explicit parametrisations.}. We also provide the poset of faces in \cite{code}.
	This will help in constructing a regular CW structure on $\mathcal{A}^{(2)}_4$ associated to its face stratification.

	\section{Unique adjoint}
	\label{sec:adjoint}

	In this section we generalize the result in \cite{Ranestad:2024svp} to the two-loop Amplituhedron, i.e. we show that the adjoint is uniquely determined by its (bi)degree and the fact that it interpolates the residual arrangement.

	\subsection{The adjoint of the loop Amplituhedron}
	
	The $L$-loop $n$-point Amplituhedron $\mathcal{A}^{(L)}_n$ for $n \geq 4$ and $L \geq 1$, is conjectured to be a (weighted) positive geometry. More precisely, we expect that there exists a unique meromorphic top-form on $\Gr(2,4)^{L}$ given by
	\begin{equation}\label{canonical_form}
		\Omega_{n}^{(L)} = \prod_{\ell=1}^{L} \langle A_{\ell}B_{\ell} dA_{\ell}dA_{\ell}\rangle \langle A_{\ell}B_{\ell} d B_{\ell}dB_{\ell}\rangle \cdot \underline{\Omega}_{n}^{(L)} \ ,
	\end{equation}
	where $\underline{\Omega}_{n}^{(L)}$ is a rational function in twistor coordinates of homogeneous degree zero in the $Z_i$'s and of degree $-4$ in each $A_{\ell}B_{\ell}$ and $dA_{\ell}$ and $dB_{\ell}$ are the differential of $A_l$ and $B_l$ respectively. The form (\ref{canonical_form}) is required to have logarithmic poles only on the boundary of $\mathcal{A}^{(L)}_n$ and therefore we can write
	\begin{equation}\label{integrand_ansatz}
		\underline{\Omega}_{n}^{(L)} = N_n^{(L)}(Z_1,\dots,Z_n)\frac{\mathcal{N}_{n}^{(L)}\bigl(A_1B_1 , \dots , A_LB_L,Z_1,\dots,Z_n\bigl)}{\prod_{\ell=1}^{L} \prod_{i=1}^{n} \bigl\langle A_{\ell}B_{\ell}  ii+1 \bigr\rangle \prod_{\ell ' < \ell} \bigl\langle A_{\ell'}B_{\ell'}  A_{\ell}B_{\ell} \bigl\rangle } \ ,
	\end{equation}
	where the numerator $\mathcal{N}_n^{(L)}$ is a homogeneous polynomial in the twistor coordinates of $A_{\ell}B_{\ell}$ of degree $n+L-1-4=n-5+L$ in each $A_{\ell}B_{\ell}$ and $Z$ matrix minors $ \langle ijkl\rangle :=\det(Z_i Z_j Z_k Z_l)$. $N_n^{(L)} $ instead is a polynomial that does not depend on twistor coordinates. We call $\mathcal{N}_n^{(L)}$ the \textit{adjoint polynomial} of $\mathcal{A}^{(L)}_n$. Note that by the symmetries of $\mathcal{A}^{(L)}_n$, it follows that $ N_n^{(L)}
	\mathcal{N}_n^{(L)}$ is invariant under the dihedral group of order $n$ acting on the $Z_i$'s, while $\mathcal{N}_n^{(L)}$ is symmetric under the action of the symmetric group $S_L$ permuting the loop variables $A_{\ell}B_{\ell}$. In particular, $\mathcal{N}_n^{(L)}$ belongs to $\mathcal{R}^{(L)}_n:=S^{L}R_{n+L-5}$,
	the symmetric $L$-fold tensor product of $R_{n+L-5}$, where $R_k$ denotes the degree-$k$ part of the coordinate ring of $\Gr(2,4)$. We denote by $d(n,L)$ the complex dimension regarded as a linear space. Then, $d(n,L)-1$ counts the degrees of freedom of the polynomial $\mathcal{N}^{(L)}_n$, up to an overall constant factor. We have that \cite[Eq. 5.1]{Ranestad:2024svp}
	\begin{equation}
		d(n,1) = \binom{n+1}{5} - \binom{n-1}{5} = 2 \binom{n}{4} - \binom{n-1}{3} \ , \quad n \geq 4 \ .
	\end{equation}
	Therefore,
	\begin{equation}
		d(n,L)= \binom{d(n+L-1,1)+L-1}{L}\;, \qquad n \geq 4 \ , \ L \geq 1 \ .
	\end{equation}
	\begin{table}[h!] 
		\centering
		\begin{tabular}{|c|c|c|c|} 
			\hline
			$n \setminus L$ &  1 & 2 & 3   \\
			\hline 
			4 & 0 & \textbf{20} & 1539  \\
			\hline 
			5 & 5 & 209 & 22099   \\
			\hline 
			6 & 19 & 1274 & 198484  \\
			\hline
		\end{tabular}
		\caption{The number of degrees of freedom of $\mathcal{N}^{(L)}_n$ up to an overall multiplicative constant, that is $d(n,L)-1$. For our case of interest, $n=4$ and $L=2$, this number is 20.}
		\label{table:adjoint_dof}
	\end{table}

	\subsection{Interpolating the residual arrangement}
	
	As promised, we now prove the following.
	
	\begin{thm}
		There exists an up to scaling unique bi-homogeneous polynomial $\mathcal{N}^{(2)}_4 \in \mathcal{R}^{(2)}_4$ of bidegree $(1,1)$ interpolating the residual arrangement of $\mathcal{A}^{(2)}_4$.
	\end{thm}
	
	\begin{proof}
		The proof is an explicit computation relying on the determination of the residual arrangement of $\mathcal{A}^{(2)}_4$ from Section \ref{sec:Real}. We begin with the general form of an element in $\mathcal{R}^{(2)}_4$, which is
		\begin{equation}\label{adjoint}
			\mathcal{N}^{(2)}_{4}(AB,CD)= \sum_{i,j=1}^{6} C_{ij} \, \langle AB Y_{i} \rangle \langle CD Y_{j} \rangle \ ,
		\end{equation}
		with
		\begin{equation}
			Y_{1}=(12) \ , \ Y_{2} = (23) \ , \ Y_{3}= (34) \ , \ Y_{4} = (41) \ , \ Y_{5}= (13) \ , \ Y_{6} = (24) \ ,
		\end{equation}
		and $C_{ij}=C_{ji} \in \mathbb{C}$.
		we show that the $20$ degrees of freedom, i.e. all coefficients $C_{ij}$ up to an overall scale, are fixed by requiring that $\mathcal{N}^{(2)}_{4}$ vanishes on all four-dimensional irreducible components of the residual arrangement, listed at the beginning of Section \ref{sec:Real}.
		
		We start from the first component $V_1^{(1)}V^{(1)}_3$ of the residual arrangement, on which
		\begin{equation}\label{first_residual_condition}
			\mathcal{N}^{(2)}_4(13,CD)= \sum_{j=1}^{6} C_{6j} \langle 1234 \rangle \langle CD Y_i \rangle \ .
		\end{equation}
		Imposing that $\mathcal{N}^{(2)}_4(13,CD)=0$ for every $CD \in \Gr(2,4)$ forces $C_{6j}=0$ for every $j=1,\dots,6$. By the dihedral and permutation symmetry, we also obtain $C_{j6}=C_{j5}=C_{5j}=0$. We proceed to the second residual component, $V_1^{(1)}P_2^{(1)}L_1^{(2)}$, on which
		\begin{equation}\label{second_residual_condition}
			\mathcal{N}^{(2)}_4(AB,CD)= \sum_{j=2}^{6} C_{3j} \langle AB34 \rangle \langle CD Y_j \rangle +  C_{6j} \langle AB24 \rangle \langle CD Y_j \rangle   \ .
		\end{equation}
		Imposing the vanishing of (\ref{second_residual_condition}) forces $C_{3j}=0$ for every $j=2,\dots,6$ in addition to the previous conditions. Again, by taking the symmetric configurations we find many more conditions, such that the only non-vanishing remaining coefficients are $C_{13}=C_{31}$ and $C_{24}=C_{42}$. At this point, one checks that $\mathcal{N}^{(2)}_4$ also interpolates the components $P^{(1)}_2V^{(2)}_2$ and $L_1^{(1)}L_3^{(1)}L_2^{(2)}L_4^{(2)}$. We are left with the last residual component $L(1,2)V_2^{(1)}P_2^{(1)}$, on which by \eqref{residual_by_drop} we have
		\begin{equation}
			\mathcal{N}^{(2)}_4(AB,CD)= (C_{13}-C_{24}) \langle AB 34 \rangle \langle CD 12 \rangle \ ,
		\end{equation}
		whose vanishing imposes $C_{24}=C_{13}=: N_4^{(2)}$. At this point, one checks that $\mathcal{N}^{(2)}_4$ vanishes also on $L(1,2)V_1^{(1)}P_2^{(1)}$, as well as on the three-dimensional residual strata of Section \ref{sec:Real}.
		
		As a result, we obtain
		\begin{equation}\label{L2n4_form}
			\begin{aligned}
				\mathcal{N}^{(2)}_4(AB,CD) = N_4^{(2)} \bigl( & \langle A B 1 2 \rangle \langle CD 34\rangle + \langle A B 34 \rangle \langle CD 12\rangle \\
				&+ \langle A B 14 \rangle \langle CD 23\rangle + \langle A B 23 \rangle \langle CD 14\rangle \bigl) \ ,
			\end{aligned}
		\end{equation}
		which for $N_4^{(2)} =\langle 1234 \rangle^{3}$ reproduces the known result, see e.g. \cite[Eq. (2)]{Dian_2023}.
	\end{proof}
	
	At this point we stress again that $\mathcal{A}^{(2)}_4$ is known to not be a positive geometry according the the definition in \cite{Arkani_Hamed_2017}, because of the presence of \textit{internal boundaries} and non-unit maximal residues \cite{Dian_2023}. Nevertheless, it is believed that $\mathcal{A}^{(2)}_4$ is a \textit{weighted positive geometry} \cite{Dian_2023}, a generalization of positive geometry. In order to prove this, one should compute the residues of (\ref{integrand_ansatz}) with the explicit form and normalization of (\ref{L2n4_form}) on every stratum of the algebraic stratification. This has been done for $\mathcal{A}^{(1)}_4$ in \cite{Ranestad:2024svp}. Some results in this direction for $\mathcal{A}^{(2)}_4$ can be found in \cite[Section 9.1]{Franco:2014csa}.

	\section*{Acknowledgements}
	The authors are grateful to Thomas Lam, Rainer Sinn, Bernd Sturmfels, and Simon Telen for their valuable discussions, and to Lizzie Pratt and Joris Köfler for their careful proofreading of the manuscript and insightful feedback.
	GD is funded by the Deutsche Forschungsgemeinschaft (DFG, German Research Foundation) under Germany’s Excellence Strategy – EXC 2121 „Quantum Universe“ – 390833306. GD also thanks the Galileo Galilei Institute for Theoretical Physics and the  (MPI-MIS) in Leipzig for the hospitality and the INFN  for partial support during the completion of this work. EM is funded by the European Union (ERC, UNIVERSE PLUS, 101118787). Views and opinions expressed are however those of the authors only and do
	not necessarily reflect those of the European Union or
	the European Research Council Executive Agency. Neither the European Union nor the granting authority can
	be held responsible for them.
	FT is funded by the Royal Society grant number
	URF\textbackslash R1\textbackslash 201473. 
	
	\clearpage

	\appendix

	\section{List of strata}
	\label{appendix_lists}
	

	\begin{table}[H]
		\centering
		\resizebox{\columnwidth}{!}{%
			\begin{tabular}{c|c|c|c|c|c}
				Schubert condition & indices &multiplicity & type & components & regions  \\
				$L_i^{(\ell)} $& &$8$ & b & $1$ & 1 \\
				\hline
				$ L(1,2)$ &  &$1$ & b & $1 $ & 1
			\end{tabular}
		}
		\caption{The $9$ boundary strata of codimension one. }
		\label{table_codim_one}
	\end{table}

	\begin{table}[H]
		\centering
		\resizebox{\columnwidth}{!}{%
			\begin{tabular}{c|c|c|c|c|c}
				Schubert condition & indices &multiplicity & type & components & regions \\
				$V_i^{(\ell)} $ &  & $8$ & b & 1 & 1 \\
				$P_i^{(\ell)} $ &  & $8$ & b & 1 & 1 \\
				$L_i^{(\ell)} L_{i+2}^{(\ell)} $&  & $4$ & b & 2 & 2 \\   
				$L_i^{(\ell)} L_j^{(\ell+1)} $ & $ j=i\pm1  $ & $8$ & b & 1 & 1\\
				$L_i^{(\ell)} L_{i+2}^{(\ell+1)} $ & & $4$ & b & 1 & 1\\
				$L_i^{(\ell)} L_i^{(\ell+1)} $ &  & $4$ & b & 1 & 2 \\
				\hline
				$ L(1,2) L_i^{(\ell)} $ &  & $8$ & b & 1 & 1 \\
			\end{tabular}
		}
		\caption{The $44$ boundary strata at codimension two.}
		\label{table_codim_two}
	\end{table}

	\begin{table}[H]
		\centering
		\resizebox{\columnwidth}{!}{%
			\begin{tabular}{c|c|c|c|c|c}
				Schubert condition & indices &multiplicity & type & components & regions   \\
				$V_i^{(\ell)}P_j^{(\ell)} $ & $|i-j|\neq 2  $ & $24$ & b & 1 & 1\\
				$V_i^{(\ell)} L_j^{(\ell+1)} $ &  & $32$ & b & 1& 1\\
				$P_i^{(\ell)} L_j^{(\ell+1)}$ &  & $32$ & b & 1& 1 \\
				$L_i^{(\ell)} L_{i+2}^{(\ell)} L^{(\ell+1)}_j $& $j=i\pm1 $ & $8$ & b & 1 & 1\\
				$L_i^{(\ell)} L_{i+2}^{(\ell)} L^{(\ell+1)}_j $& $j=i,i+2 $ & $8$ & b & 1 & 2\\         
				\hline
				$L(1,2) V_i^{(\ell)} $ &  & $8$ & b & 1& 1 \\
				$L(1,2) P_i^{(\ell)} $ &  & $8$ & b & 1& 1 \\
				$L(1,2) L_i^{(\ell)} L_{i+2}^{(\ell)} $ &  & $4$ & b & 2 & 2 \\
				$L(1,2) L_i^{(\ell)}L_{j}^{(\ell+1)} $ & $j = i\pm1 $ & $8$ & b & 1& 1 \\
				$L(1,2) L_i^{(\ell)}L_{i+2}^{(\ell+1)} $ &  & $4$ & b & 1 & 4 \\
				$V_i(1,2)  $ &  & $4$ & b & 1 & 2 \\
				$P_i(1,2) $ &  & $4$ & b & 1 & 2 \\
			\end{tabular}
		}
		\caption{The $144$ boundary strata at codimension three.}
		\label{table_codim_three}
	\end{table}

	\begin{table}[H]
		\centering
		\resizebox{\columnwidth}{!}{%
			\begin{tabular}{c|c|c|c|c|c}
				Schubert condition & indices &multiplicity& type & components & regions  \\
				$V_i^{(\ell)}V_j^{(\ell)} $ & $j=i\pm1 $ & $8$ & b & 1& 1 \\
				$V_i^{(\ell)}V_j^{(\ell)} $ & $j=i,i+2 $ & $4$ & r &  \\
				$V_i^{(\ell)} P_j^{(\ell)} L_i^{(\ell+1)} $ & $j=i\pm1 $ & $16$ & r &  \\
				$V_i^{(\ell)} P_j^{(\ell)} L_k^{(\ell+1)} $ & else & $80$ & b & 1  & 1\\
				$V_i^{(\ell)} V_j^{(\ell+1)}$ & $i \neq j  $ & $12$ & b & 1 & 1\\
				$V_i^{(\ell)} V_i^{(\ell+1)}$ &  & $4$ & b & 1 & 4 \\
				$V_i^{(\ell)} P_j^{(\ell+1)}$ & $i \neq j$ & $24$ & b & 1& 1 \\
				$V_i^{(\ell)} P_i^{(\ell+1)}$ &  & $8$ & r & \\
				$V_i^{(\ell)} L_{j}^{(\ell+1)} L^{(\ell+1)}_{j+2} $&  & $16$ & b & 1 & 1\\         
				$P_i^{(\ell)} P_j^{(\ell+1)}$ & $i \neq j $ & $12$ & b & 1 & 1\\
				$P_i^{(\ell)} P_i^{(\ell+1)}$ &  & $4$ & b & 1 & 4 \\
				$P^{(\ell)}_i L_j^{(\ell+1)}L_{j+2}^{(\ell+1)} $ &  & $16$ & b & 1& 1 \\
				$L_i^{(\ell)} L_{i+2}^{(\ell)} L_{i-1}^{(\ell+1)} L^{(\ell+1)}_{i+1} $&  & $2$ & r & \\ 
				$L_i^{(\ell)} L_{i+2}^{(\ell)} L_{i}^{(\ell+1)} L^{(\ell+1)}_{i+2} $&  & $2$ & b & 1 & 2 \\
				\hline
				$L(1,2) V_{i}^{(\ell)} P^{(\ell)}_{j} $& $j=i \pm 1$ & $16$ & b & 1  & 1\\ 
				$L(1,2) V_{i}^{(\ell)} P^{(\ell)}_{i} $&  & $8$ & r &   \\ 
				$L(1,2) V_{i}^{(\ell)} L^{(\ell+1)}_{j} $& $j=i+1,i+2$ & $16$ & b & 1 & 1 \\
				$P_j(1,2) V_{i}^{(\ell)} $& $j=i-1,i$ & $16$ & b & 1 & 1\\
				$L(1,2) P_{i}^{(\ell)} L^{(\ell+1)}_{j} $& $j = i+1,i+2 $ & $16$ & b & 1& 1 \\
				$V_j(1,2) P^{(\ell)}_{i}  $& $j=i-1,i $ & $16$ & b & 1& 1 \\
				$L(1,2) L_{i}^{(\ell)} L_{i+2}^{(\ell)} L^{(\ell+1)}_{j} $& $j=i-1,i+1 $ & $8$ & b & 1& 1 \\
				$V_i(1,2) P_i(1,2) $&  & $4$ & b & 1 & 1\\
				$V_i(1,2) L_{j}^{(\ell)} L_{j+2}^{(\ell)}  $& $j=i,i+2 $ & $8$ & b & 1 & 2 \\
				$P_i(1,2) L_{j}^{(\ell)} L_{j+2}^{(\ell)} L^{(\ell+1)}_{i} $& $j=i,i+2 $ & $8$ & b & 1 & 2 \\
			\end{tabular}
		}
		\caption{The $324$ boundary strata at codimension four. }
		\label{table_codim_four}
	\end{table}

	\begin{table}[H]
		\centering
		\resizebox{\columnwidth}{!}{%
			\begin{tabular}{c|c|c|c|c|c}
				Schubert condition & indices &multiplicity & type & components & regions \\
				$V_i^{(\ell)}V_j^{(\ell)}L^{(\ell+1)}_k $ & $j=i \pm 1$ & $32$ & b & 1 & 1 \\
				$V_i^{(\ell)}V_j^{(\ell)}L^{(\ell+1)}_k $ & $j=i,i+2$ & $16$ & r & \\
				$V_i^{(\ell)} P_{j}^{(\ell)} V_{j}^{(\ell+1)} $ & $j=i\pm1$ & $16$ & r & \\
				$V_i^{(\ell)} P_{j}^{(\ell)} V_i^{(\ell+1)} $ & $j=i\pm1$ & $16$ & b & 1 & 2 \\
				$V_i^{(\ell)} P_j^{(\ell)} V_k^{(\ell+1)} $ & else & $64$ & b & 1 & 1 \\
				$V_i^{(\ell)} P_j^{(\ell)} P_i^{(\ell+1)} $ & $ j = i \pm 1 $ & $16$ & r & \\
				$V_i^{(\ell)} P_j^{(\ell)} P_j^{(\ell+1)} $ & $j = i \pm 1 $ & $16$ & b & 1 & 2 \\
				$V_i^{(\ell)} P_j^{(\ell)} P_k^{(\ell+1)} $ & else & $64$ & b & 1 & 1 \\
				$V_i^{(\ell)} P_{i+1}^{(\ell)} L_i^{(\ell+1)} L_{i+2}^{(\ell+1)} $ &  & $16$ & r & \\     
				$V_i^{(\ell)} P_j^{(\ell)} L_k^{(\ell+1)} L_{k+2}^{(\ell+1)} $ & else & $32$ & b & 1 & 1 \\     
				\hline
				$L(1,2) V_{i}^{(\ell)} V^{(\ell)}_{i+2}  $&  & $4$ & r & \\ 
				$ L(1,2) V_{i}^{(\ell)} P^{(\ell)}_{j}L^{(\ell+1)}_{j+1} $& $j=i\pm 1$ & $16$ & b & 1 & 1 \\
				$ L(1,2) V_{i}^{(\ell)} P^{(\ell)}_{i}L^{(\ell+1)}_k $& $k=i+1,i+2$ & $16$ & r & \\
				$ P_{j+2}(1,2) V_{i}^{(\ell)} P^{(\ell)}_{j} $& $j= i \pm 1$ & $16$ & b & 1 & 1 \\
				$ V_i(1,2) V_{i}^{(\ell)} P^{(\ell)}_{j} $& $j= i \pm 1$ & $16$ & b & 1 & 1 \\
				$ V_i(1,2) L_{j+2}^{(\ell)} L_{i}^{(\ell+1)} L_{i+2}^{(\ell+1)}  $& $j=i,i+2 $ & $8$ & b & 1 & 2 \\
				$P_j(1,2) V_{i}^{(\ell)} V^{(\ell+1)}_{i} $& $j=i\pm 1$ & $8$ & b & 1 & 2 \\
				$L(1,2) V_{i}^{(\ell)} V^{(\ell+1)}_{i+2} $&  & $4$ & r & \\
				$P_i(1,2) V_{i}^{(\ell)} V^{(\ell+1)}_{j} $& $j=i\pm 1$ & $8$ & b & 1 & 1 \\
				$L(1,2) V_{i}^{(\ell)} P^{(\ell+1)}_{i+2} $&  & $8$ & b & 1 & 1\\
				$P_i(1,2) V_{i}^{(\ell)} L^{(\ell+1)}_{j} L^{(\ell+1)}_{j+2} $& $i,j \in [4]$ & $16$ & b & 1 & 1\\
				$V_i(1,2) P_{j}^{(\ell)} P^{(\ell+1)}_{j}  $& $j=i,i-1$ & $8$ & b & 1 & 2 \\
				$V_i(1,2) P_{i}^{(\ell)} P^{(\ell+1)}_{i+1}  $& & $8$ & b & 1 & 1\\
				$L(1,2) P_{i}^{(\ell)} P^{(\ell+1)}_{i+2}  $&  & $4$ & r & \\
				$V_{i}(1,2) P_{j}^{(\ell)} L^{(\ell+1)}_{i}  L^{(\ell+1)}_{i+2}  $& $j=i,i+1$ & $16$ & b & 1 & 1\\
				$V_{i}(1,2) P_{i+2}(1,2)  $&  & $4$ & b & 1 & 2 \\
				$L(1,2) L_{i}^{(\ell)} L_{i+2}^{(\ell)} L^{(\ell+1)}_{i-1}  L^{(\ell+1)}_{i+1}  $&  & $2$ & r &
			\end{tabular}
		}
		\caption{The $450$ boundary strata at codimension five. }
		\label{table_codim_five}
	\end{table}

	\begin{table}[H]
		\centering
		\resizebox{\columnwidth}{!}{%
			\begin{tabular}{c|c|c|c|c|c}
				Schubert condition & indices &multiplicity& type & components & regions  \\
				$V_i^{(\ell)}P_{i+1}^{(\ell)} V_{j}^{(\ell+1)}P_{j+1}^{(\ell+1)} $ & $j=i\pm1$ & $24$ & r &  \\
				$V_i^{(\ell)}P_j^{(\ell)} V_i^{(\ell+1)}P_j^{(\ell+1)} $ & $|i-j|\neq 2 $ & $12$ & b & 1 & 2 \\
				$V_i^{(\ell)}P_j^{(\ell)} V_k^{(\ell+1)}P_l^{(\ell+1)} $ & else & $116$ & b & 1 & 1 \\
				$V_i^{(\ell)} V_{i+2}^{(\ell)} V_j^{(\ell+1)} $ & $j=i \pm 1$ & $8$ & r & \\
				$V_i^{(\ell)} V_j^{(\ell)} V_k^{(\ell+1)} $ & else, $i\neq j$ & $40$ & b & 1 & 1 \\
				$V_i^{(\ell)} V_{i+2}^{(\ell)} P_j^{(\ell+1)} $ & $j=i,i+2$ & $8$ & r & \\
				$V_i^{(\ell)} V_j^{(\ell)} P_k^{(\ell+1)} $ & else, $i \neq j$ & $40$ & b & 1 & 1 \\
				$V_i^{(\ell)} V_{i+2}^{(\ell)} L_j^{(\ell+1)} L_{j+2}^{(\ell+1)} $ &  & $8$ & r & \\
				$V_i^{(\ell)} V_{i+1}^{(\ell)} L_j^{(\ell+1)} L_{j+2}^{(\ell+1)} $ &  & $16$ & b & 1 & 1 \\
				\hline
				$L(1,2) V_i^{(\ell)}P_i^{(\ell)} V_{i+2}^{(\ell+1)} $ &  & $8$ & r & \\
				$L(1,2) V_i^{(\ell)}P_{j}^{(\ell)} V_{j+2}^{(\ell+1)} $ & $j=i\pm 1 $ & $16$ & b & 1 & 1 \\
				$P_{j+2}(1,2) V_i^{(\ell)}P_{j}^{(\ell)} V_{i}^{(\ell+1)} $ & $j=i \pm 1 $ & $16$ & b & 1 & 2 \\
				$L(1,2) V_i^{(\ell)} P^{(\ell)}_i P_{i+2}^{(\ell+1)}  $ &  & $8$ & r &  \\
				$V_i(1,2) V_i^{(\ell)} P_{j}^{(\ell)} P_{i+2}^{(\ell+1)}  $ & $j=i \pm 1$ & $32$ & b & 1& 1 \\
				$V_i(1,2) V_i^{(\ell)} P_{j}^{(\ell)} L_{j}^{(\ell+1)} L_{j+2}^{(\ell+1)}  $ & $j = i \pm 1 $ & $16$ & b & 1& 1 \\
				\hline
				$V^{(12)}_{i}  $ &  & $4$ & b & 1 & 1\\
				$P^{(12)}_{i}$ &  & $4$ & b & 1 & 1\\
				$V_i(1,2) V_{i+2}(1,2)   $ &  & $2$ & b & 1 & 1\\     
			\end{tabular}
		}
		\caption{The $370$ boundary strata at codimension six. }
		\label{table_codim_six}
	\end{table}

	\begin{table}[H]
		\centering
		\resizebox{\columnwidth}{!}{%
			\begin{tabular}{c|c|c|c|c|c}
				Schubert condition & indices &multiplicity& type & components & regions \\
				$V_i^{(\ell)}V_{i+2}^{(\ell)}V^{(\ell+1)}_j P^{(\ell+1)}_i $ & $j=i \pm 1 $ & $8$ & r & \\
				$V_i^{(\ell)}V_j^{(\ell)}V^{(\ell+1)}_kP^{(\ell+1)}_l $ & else, $|k-l|\neq 2$ & $136$ & b & 1 & 1\\
				\hline
				$L(1,2) V_i^{(\ell)} P_i^{(\ell)} V^{(\ell+1)}_{i+2}P^{(\ell+1)}_{i+2} $ &  & $4$ & r &  \\
				$V_i(1,2)V_{i-1}^{(\ell)}V^{(\ell+1)}_{i+2} $ &  & $8$ & b & 1 & 1\\
				\hline
				$V_i^{(12)} P_j^{(\ell)} P^{(\ell+1)}_j$ & $j=i \pm 1 $ & $8$ & b & 1 & 1 \\
				$V_i^{(12)} P^{(12)}_i  $ &  & $4$ & b & 1 & 1\\
				
			\end{tabular}
		}
		\caption{The $168$ boundary strata at codimension seven. }
		\label{table_codim_seven}
	\end{table}

	\begin{table}[H]
		\centering
		\resizebox{\columnwidth}{!}{%
			\begin{tabular}{c|c|c|c|c|c}
				Schubert condition & indices &multiplicity & type & components & regions \\
				$V_i^{(\ell)}V_{i+2}^{(\ell)}V_j^{(\ell+1)}V_{j+2}^{(\ell+1)} $ & $j=i+1 $ & $2$ & r & \\
				$V_i^{(\ell)}V_j^{(\ell)}V_k^{(\ell+1)}V_l^{(\ell+1)} $ & else & $34$ & b & 1 & 1 \\
			\end{tabular}
		}
		\caption{The $36$ boundary strata at codimension eight.}
		\label{table_codim_eight}
	\end{table}

	\bibliography{bibliography}

\begin{thebibliography}{10}

\bibitem{Arkani_Hamed_2017}
N.~Arkani-Hamed, Y.~Bai, and T.~Lam.
\newblock Positive geometries and canonical forms.
\newblock {\em Journal of High Energy Physics}, 2017(11), Nov. 2017.

\bibitem{Arkani_Hamed_2018}
N.~Arkani-Hamed, H.~Thomas, and J.~Trnka.
\newblock Unwinding the amplituhedron in binary.
\newblock {\em Journal of High Energy Physics}, 2018(1), Jan. 2018.

\bibitem{Arkani-Hamed:2013kca}
N.~Arkani-Hamed and J.~Trnka.
\newblock {Into the Amplituhedron}.
\newblock {\em JHEP}, 12:182, 2014.

\bibitem{Chen2013}
C.~Chen, J.~H. Davenport, J.~P. May, M.~Moreno~Maza, B.~Xia, and R.~Xiao.
\newblock Triangular decomposition of semi-algebraic systems.
\newblock {\em J. Symbolic Comput.}, 49:3--26, 2013.

\bibitem{Collins1975}
G.~E. Collins.
\newblock Quantifier elimination for real closed fields by cylindrical algebraic decomposition.
\newblock In {\em Automata theory and formal languages ({S}econd {GI} {C}onf., {K}aiserslautern, 1975)}, volume Vol. 33 of {\em Lecture Notes in Comput. Sci.}, pages 134--183. Springer, Berlin-New York, 1975.

\bibitem{noauthororeditor}
D.~Cox, J.~Little, and D.~O’Shea.
\newblock Ideals, varieties, and algorithms. an introduction to computational algebraic geometry and commutative algebra.
\newblock 2007.

\bibitem{GabriBlog}
G.~Dian.
\newblock On the polygon front lines: visualizing the amplituhedron with the wolfram language.
\newblock \url{https://shorturl.at/vgGwc}.

\bibitem{Dian_2023}
G.~Dian, P.~Heslop, and A.~Stewart.
\newblock Internal boundaries of the loop amplituhedron.
\newblock {\em SciPost Physics}, 15(3), Sept. 2023.

\bibitem{code}
G.~Dian, E.~Mazzucchelli, and F.~Tellander.
\newblock 2-loop amplituhedron stratification.
\newblock \url{https://github.com/GabrieleDian/Loop-amplituhedron-stratification.git}.

\bibitem{Even-Zohar:2023del}
C.~Even-Zohar, T.~Lakrec, M.~Parisi, R.~Tessler, M.~Sherman-Bennett, and L.~Williams.
\newblock {Cluster algebras and tilings for the m=4 amplituhedron}.
\newblock 10 2023.

\bibitem{Franco:2014csa}
S.~Franco, D.~Galloni, A.~Mariotti, and J.~Trnka.
\newblock {Anatomy of the Amplituhedron}.
\newblock {\em JHEP}, 03:128, 2015.

\bibitem{Galashin_2022}
P.~Galashin, S.~N. Karp, and T.~Lam.
\newblock The totally nonnegative grassmannian is a ball.
\newblock {\em Advances in Mathematics}, 397:108123, Mar. 2022.

\bibitem{M2}
D.~R. Grayson and M.~E. Stillman.
\newblock {Macaulay2, a software system for research in algebraic geometry}.
\newblock Available at \url{http://www.math.uiuc.edu/Macaulay2}.

\bibitem{kohn2021}
K.~Kohn, R.~Piene, K.~Ranestad, F.~Rydell, B.~Shapiro, R.~Sinn, M.-S. Sorea, and S.~Telen.
\newblock Adjoints and canonical forms of polypols.
\newblock 2021.

\bibitem{Lam:2022yly}
T.~Lam.
\newblock {An invitation to positive geometries}.
\newblock 8 2022.

\bibitem{lemaire2005regularchains}
F.~Lemaire, M.~M. Maza, and Y.~Xie.
\newblock The \texttt{RegularChains} library in \texttt{MAPLE}.
\newblock {\em ACM SIGSAM Bulletin}, 39(3):96--97, 2005.

\bibitem{Lukowski:2020bya}
T.~\L{}ukowski and R.~Moerman.
\newblock {Boundaries of the amplituhedron with amplituhedronBoundaries}.
\newblock {\em Comput. Phys. Commun.}, 259:107653, 2021.

\bibitem{postnikov2006}
A.~Postnikov.
\newblock Total positivity, grassmannians, and networks.
\newblock 2006.

\bibitem{Ranestad:2024svp}
K.~Ranestad, R.~Sinn, and S.~Telen.
\newblock {Adjoints and Canonical Forms of Tree Amplituhedra}.
\newblock 2 2024.

\end{thebibliography}

\end{document}